\documentclass{llncs}

\usepackage[utf8]{inputenc}
\usepackage[T1]{fontenc}
\usepackage{enumerate}
\usepackage{graphicx}
\usepackage{amsfonts}
\usepackage{amssymb}
\usepackage{amsmath}
\usepackage{wasysym}
\usepackage{hyperref}
\usepackage{mathtools}
\usepackage{bibentry}
\usepackage{comment}
\usepackage{paralist}
\usepackage{enumitem}
\usepackage{todonotes}
\usepackage{hyperref}
\usepackage{thmtools}
\usepackage{float}
\usepackage{lineno}

\definecolor{red}{HTML}{000000}
\definecolor{blue}{HTML}{000000}

\graphicspath{{./figures/}}

\title{List covering of regular multigraphs with semi-edges\thanks{The conference version appeared at IWOCA 2022~\cite{iwoca}. The first author is the corresponding author.}}

\author{Jan Bok\inst{1}\orcidID{0000-0002-7973-1361} \and Jiří Fiala\inst{2}\orcidID{0000-0002-8108-567X} \and Nikola Jedličková\inst{2}\orcidID{0000-0001-9518-6386} 
\and Jan Kratochvíl\inst{2}\orcidID{0000-0002-2620-6133} 
\and Pawe\l{} Rz\k{a}\.zewski\inst{3,4}\orcidID{0000-0001-7696-3848}
}

\institute{
Computer Science Institute, Faculty of Mathematics and Physics, Charles University, Prague, Czech Republic, \url{bok@iuuk.mff.cuni.cz}
\and
Department of Applied Mathematics, Faculty of Mathematics and Physics, Charles University, Prague, Czech Republic, \url{fiala,jedlickova,honza@kam.mff.cuni.cz}
\and
Warsaw University of Technology, Warsaw, Poland, \url{pawel.rzazewski@pw.edu.pl}
\and
University of Warsaw, Warsaw, Poland
}

\begin{document}

\maketitle

\begin{abstract}
In line with the recent development in topological graph theory, we are considering undirected graphs that are allowed to contain {\em multiple edges}, {\em loops}, and {\em semi-edges}. A graph is called {\em simple} if it contains no semi-edges, no loops, and no multiple edges.
A graph covering projection, also known as a locally bijective homomorphism, is a mapping between vertices and edges of two graphs which preserves incidences and which is a local bijection on the edge-neighborhood of every vertex. This notion stems from topological graph theory, but has also found applications in combinatorics and theoretical computer science.

It has been known that for every fixed simple regular graph $H$ of valency greater than 2, deciding if an input graph covers $H$ is NP-complete.  Graphs with semi-edges have been considered in this context only recently and only partial results on the complexity of covering such graphs are known so far. In this paper we consider the list version of the problem, called \textsc{List-$H$-Cover}, where the vertices and edges of the input graph come with lists of admissible targets. Our main result reads that the \textsc{List-$H$-Cover} problem is NP-complete for every regular graph $H$ of valency greater than 2 which contains at least one semi-simple vertex (i.e., a vertex which is incident with no loops, with no multiple edges and with at most one semi-edge).
Using this result we show the NP-co/polytime dichotomy for the computational complexity of \textsc{ List-$H$-Cover} for cubic graphs.
\end{abstract}

\section{Introduction}\label{sec:Intro}
\noindent\emph{Graphs.} In this paper, we consider undirected graphs in the most relaxed form of the definition -- they are allowed to contain multiple edges, loops and semi-edges. This is in line with the current development of topological graph theory, where loops, multiple edges, and 
semi-edges have become commonly accepted. Intuitively, a {\em semi-edge} (sometimes also called a \emph{half-edge} or a \emph{fin})
is an edge with just one end, in contrast with a {\em loop}, which has two ends, both in the same vertex. An {\em ordinary edge} has two ends, each in a different vertex. {\em Multiple edges} are edges incident with the same vertex (in case of loops and semi-edges), or the same pair of vertices (in case of ordinary edges). A graph is {\em simple} if it contains no loops, no semi-edges, and no multiple edges. For a formal definition, cf. Section~\ref{sec:prelim}.

\noindent\emph{Graph covering projections and related notions.}
For simple graphs $G$ and $H$, a covering projection from $G$ to $H$ is a mapping $f : V(G)  \to V(H)$,
such that (i) $f$ is  adjacency preserving (i.e., it is a graph homomorphism) and (ii) $f$ is bijective in the neighborhood of each vertex. The last condition means that if for some $v \in V(G)$ and $x \in V(H)$ we have $f(v) = x$, then for each neighbor $y$ of $x$ in $H$, there must be exactly one neighbor $w$ of $v$ in $G$ that is mapped to $y$.
For a fixed graph $H$, in the \textsc{$H$-Cover} problem we ask if an input graph $G$ admits a covering projection to $H$. A covering projection from $G$ to $H$ is also referred to as a {\em locally bijective homomorphism} from $G$ to $H$. For graphs with loops, multiple edges and semi-edges, the covering projection is defined both on vertices and edges as an incidence preserving mapping which is bijective on edge-neighborhoods of all vertices. For a formal definition, cf. Section~\ref{sec:prelim}.

The notion of a graph covering projection, as a natural discretization of the covering projection used in topology, originates (not surprisingly) in topological graph theory. However, since then it has found numerous applications elsewhere.
Covering projections were used for constructing highly symmetrical graphs~\cite{k:Biggs74,n:Djokovic74,n:Gardiner74,n:GT77}, embedding complete graphs in surfaces of higher genus~\cite{k:Ringel74}, and for analyzing a model of local computations~\cite{n:Angluin80}.

Graph covering projections are also known as \emph{locally bijective homomorphisms} and as such they fall into a family of \emph{locally constrained homomorphisms}.
Other problems from this family are locally surjective and locally injective graph homomorphisms, where we ask for the existence of a homomorphism that is, respectively, surjective or injective in the neighborhood of each vertex.
Locally surjective homomorphisms play an important role in social sciences~\cite{n:FP05} (there this problem is called the Role Assignment Problem).
On the other hand, a prominent special case of the locally injective homomorphism problem is the well-studied $L(2,1)$-labeling problem~\cite{n:GriggsY92}, and, more generally, $H(p,q)$-coloring~\cite{n:FHKT03,n:KT00}

\smallskip
\noindent\emph{Computational complexity.}
The complexity of finding locally constrained homomorphisms was studied by many authors.
For locally surjective homomorphisms we know a complete dichotomy~\cite{n:FP05}. The problem is polynomial-time solvable if the target graph $H$ either (a) has no edge, or (b) has a component that consists of a single vertex with a loop, or (c) is simple and bipartite, with at least one component isomorphic to $K_2$. In all other cases the problem is NP-complete.

The dichotomy for locally injective homomorphisms is still unknown, despite some work~\cite{DBLP:conf/tamc/BilkaLT11,n:FK01,DBLP:journals/dam/FialaKP08,n:FK02,DBLP:conf/iwoca/LidickyT10}. However, we understand the complexity of the \emph{list} variant of the problem~\cite{n:FK06}: it is polynomial-time solvable if every component of the target graph has at most one cycle, and NP-complete otherwise. 

To the best of our knowledge, Abello et al.~\cite{n:AFS91} were the first to ask about the computational complexity of {\sc $H$-Cover}.
Note that in order to map a vertex of $G$ to a vertex of $H$, they must be of the same degree; a natural interesting special case is when $H$ is regular. It is known that for every $k \geq 3$, the \textsc{$H$-Cover} problem is NP-complete for every simple $k$-regular graph $H$~\cite{n:KPT97,n:FK08}. (For $k\le 2$, \textsc{$H$-Cover} is polynomial time solvable for $k$-regular $H$ - for $k=1$, the only cover of $K_2$ is $K_2$ itself, and for $k=2$, the cycle of length $t$ is covered only by cycles whose lengths are multiples of $t$.)
Some other partial results are known, mostly focusing on small graphs $H$~\cite{n:Fiala00b,n:KPT98,n:KratochvilTT16}.
Let us point out that in all the above results it was assumed that $H$ has no multiple edges, no loops and no semi-edges.

Recall further that there is also some more work concerning the complexity of locally surjective and injective homomorphisms if $G$ is assumed to come from some special class of graphs~\cite{n:BardBDMY18,n:BilkaJKTV11,n:ChaplickFHPT15,DBLP:journals/corr/abs-2202-12438,n:FialaKKN14,n:OkrasaR20}. We also refer the reader to the survey concerning various aspects of locally constrained homomorphisms~\cite{n:FK08}.

\smallskip
\noindent\emph{The role of semi-edges.}
Let us name just a few of the most significant examples of usage of semi-edges. Malni{\v{c}} et al.~\cite{n:MalnivcMP04} considered
semi-edges during their study of abelian covers to allow for a broader range of applications. 
Furthermore, the concept of graphs with semi-edges was introduced independently and
naturally in mathematical physics~\cite{getzler1998modular}.
It is also natural to consider semi-edges in the mentioned framework of local computations (we refer to the introductory section of~\cite{n:BFHJK-MFCS} for more details). Finally, a theorem of Leighton~\cite{n:Leighton82} on finite common covers has been recently generalized to the semi-edge setting in~\cite{arxiv1908.00830,woodhouse_2021}. To highlight a few other contributions, the reader is invited
to consult~\cite{n:MalnicNS00,n:NedelaS96}, the surveys~\cite{kwak2007graphs,nedela_mednykh}, and
finally for more recent results, the series of papers~\cite{n:FialaKKN14,n:FialaKKN18,arxiv1609.03013} and the introductions therein. 

The complexity study of \textsc{$H$-Cover} for graphs $H$ that allow semi-edges has been initiated only very recently in~\cite{n:BFHJK-MFCS,BFJKS21FCT}. We continue this line of research. In particular, our far-reaching goal is to prove the following conjecture.

\begin{conjecture}[Strong Dichotomy Conjecture] \label{conj:sdc}
For every $H$, the {\sc $H$-Cover} 
problem is either  polynomial-time solvable for general graphs on input, or NP-complete for simple input graphs.
\end{conjecture}

\noindent\emph{Our results.}
The goal of this paper is to push further the understanding of the complexity of  \textsc{$H$-Cover}  for regular graphs. Recall that the problem is known to be NP-complete for every fixed $k$-regular \emph{simple} graph $H$ of valency $k\ge 3$~\cite{n:KPT97,n:Fiala00b,n:FK08}. 
Though it was known already from~\cite{n:KPT97a} that in order to fully understand the complexity of covering general simple graphs, it is necessary (and sufficient) to prove a complete characterization for colored mixed graphs with loops and multiple edges allowed (but no semi-edges), the result of~\cite{n:KPT97} was formulated and proved only for simple graphs.
In this paper we revisit the method developed in~\cite{n:KPT97} and we conclude that though it does not seem to work for graphs with multiple edges in general, it is possible to modify it and -- under certain assumptions -- prove NP-hardness for the \emph{list} variant of the problem, \textsc{List-$H$-Cover}, where the vertices and edges of the input graph are given with lists of admissible targets. (Note that for any graph $H$, \textsc{List-$H$-Cover} is at least as difficult as \textsc{$H$-Cover}, since the latter problem is a special instance of the former one with all lists full, i.e., containing all vertices or edges of the target graph $H$.)
One of our main results is the following theorem (a vertex is {\em semi-simple} if it belongs to no loops and no multiple edges, and is incident with  at most one semi-edge).

\begin{theorem}
\label{thm:mainresult}
Let $k \geq 3$ and let $H$ be a connected $k$-regular graph with at least one semi-simple vertex. Then {\sc List-$H$-Cover} is NP-complete for simple input graphs.
\end{theorem}
We do believe that the Strong Dichotomy Conjecture holds true also for  \textsc{List-$H$-Cover}. \textcolor{blue}{Another observation in support of this conjecture is the case of regular graphs of lower valency.}

\textcolor{blue}{
\begin{theorem}
\label{thm:minorresult}
Let $k \le 2$ and let $H$ be a connected $k$-regular graph.
Then {\sc List-$H$-Cover} is polynomial-time solvable for arbitrary input graphs.
\end{theorem}
}

The second main goal of the current paper is to show how Theorem~\ref{thm:mainresult} can be used to prove the Strong Dichotomy Conjecture for cubic graphs.
Recall that for the closely related locally injective homomorphism problem, introducing lists was helpful in obtaining the full complexity dichotomy~\cite{n:FK06}. For locally surjective homomorphisms, it follows straightforwardly from Theorem~3 of~\cite{n:FP05} that the list version is NP-complete whenever the target graph contains a connected component with at least two edges. On the other hand, if the connected components of the target graph have at most one edge each, the list version is still polynomial time solvable, one only has to check that all vertices of each component of the source graph are allowed to be mapped onto the vertices of the same component of the target one.

\begin{theorem}\label{thm:cubicdichotomy}
Let $H$ be a connected cubic graph. Then {\sc List-$H$-Cover} is polynomial-time solvable for  general graphs when $H$ has one vertex and one semi-edge, and
it is NP-complete (with respect to Turing reductions) even for simple input graphs otherwise.
\end{theorem}

The paper is organized as follows. In Section~\ref{sec:prelim} we present formal definitions of graphs and graph covers, and of further notions used later in the paper. In Section~\ref{sec:mainresult}, we prove Theorem~\ref{thm:mainresult} in two steps -- first for bipartite graphs $H$, and then for arbitrary ones. Results on the computational complexity of the {\sc $H$-Cover} and {\sc List-$H$-Cover} problems for several special cubic graphs, which are interesting on their own, are proved in Section~\ref{sec:sausages}, and the Strong Dichotomy for regular subcubic graphs is proved in Section~\ref{sec:cubic}. Final remarks are collected in Section~\ref{sec:concl}.    

\section{Preliminaries}\label{sec:prelim}

In this section we present formal definitions of graphs and graph covering projections, as well as some auxiliary notions used later in the paper.

\begin{definition}\label{def:graph}
A {\em graph} is a triple $G=(V,\Lambda,\iota)$ where $V$ is its set of {\em vertices}, $\Lambda=E\cup L\cup S$ is its set of {\em edges} and $\iota:\Lambda\to V\cup {V\choose{2}}$ is the incidence function. Here $E$ is the set of {\em ordinary edges} and $\iota(e)\in {V\choose{2}}$ for every $e\in E$, $L$ is the set of {\em loops} and $S$ is the set of {\em semi-edges} with $\iota(e)\in V$ for every $e\in L\cup S$.   
For a graph $G=(V,\Lambda,\iota)$, by $V(G)$ we denote the set $V$. Similarly we use $\Lambda(G)$ and $E(G)$).
\end{definition}

We say that a vertex $u$ is {\em incident with} an edge $e$ (and vice-versa) if $v\in \iota(e)$ or $v = \iota(e)$. Given a graph $G$ and a vertex $u\in V(G)$, the set of edges of $G$ incident with $u$ will be denoted by $\Lambda_G(u)$. The difference between loops and semi-edges lies in their contributions to the degrees of the vertices they are incident with. Informally, while an ordinary edge has two distinct end-vertices, a loop has two end-vertices which are identical, and a semi-edge has just one end-vertex.  
The \emph{degree} (or \emph{valency}) of a vertex $u$ is the number of edge end-vertices equal to $u$.
In particular, each ordinary edge and each semi-edge contribute 1 to the degree of each of its vertices, and each loop contributes 2.
A graph is {\em $k$-regular} if all of its  vertices have the same degree $k$.
We further say that:
\begin{compactitem}
\item a vertex is {\em semi-simple} if it belongs to no loops, no multiple edges and at most one semi-edge,
\item a graph is {\em semi-simple} if each of its vertices is semi-simple,
\item a vertex is {\em simple} if it is semi-simple and is incident with no semi-edges,
\item a graph is {\em simple} if each of its vertices is simple,
\item a graph is {\em bipartite} if it has no loops, no semi-edges, and no odd cycles.
\end{compactitem}

\begin{definition}\label{def:cover}
Given graphs $G$ and $H$, a mapping $f:V(G)\cup \Lambda(G)\longrightarrow V(H)\cup \Lambda(H)$ is a {\em graph covering projection} if vertices of $G$ are mapped onto vertices of $H$, edges of $G$ are mapped onto edges of $H$ so that incidences are retained, and in such a way that the preimage of a loop is a disjoint union of cycles spanning the preimage of the vertex incident with the loop (note that a loop itself is a cycle of length 1), the preimage of a semi-edge is a disjoint union of semi-edges and ordinary edges spanning the preimage of the vertex incident with this semi-edge, and the preimage of an ordinary edge is a matching spanning the preimage of the two vertices incident with this edge.
\end{definition}

The computational problem of deciding whether an input graph $G$ covers a fixed graph $H$ is denoted by {\sc $H$-Cover}.

The mapping  $f:V(G)\cup \Lambda(G)\longrightarrow V(H)\cup \Lambda(H)$ is a {\em partial covering projection} when the preimages are not required to be spanning subgraphs, but all other properties are fulfilled. In other words, the vertex- and edge-mappings are both surjective and the incidences are retained, the preimage of an ordinary edge connecting vertices say $u$ and $v$ is a matching consisting of edges each connecting a vertex from $f^{-1}(u)$ to a vertex from $f^{-1}(v)$, the preimage of a semi-edge incident with vertex say $u$ is a disjoint union of semi-edges and ordinary edges all incident only with vertices from  $f^{-1}(u)$, and the preimage of a loop incident with a vertex say $u$ is a disjoint union of cycles (including loops) and paths whose all edges are incident only with vertices from $f^{-1}(u)$.

In the {\sc List-$H$-Cover} problem the input graph $G$ is given with lists ${\cal L}=\{L_u,L_e: u\in V(G), e\in \Lambda(G)\}$, such that $L_u\subseteq V(H)$ for every $u\in V(G)$ and $L_e \subseteq \Lambda(H)$ for every $e\in E(G)$. A covering projection $f:G\longrightarrow H$ {\em respects} the lists of $\cal L$ if $f(u)\in L_u$ for every $u\in V(G)$ and $f(e)\in L_e$ for every $e\in \Lambda(G)$. 

\textcolor{red}{
\section{Complexity of {\sc List-$H$-Cover} for graphs with semi-simple vertices}}\label{sec:mainresult}

\textcolor{red}{The goal of this section is to prove Theorem~\ref{thm:mainresult}. The {\sc List-$H$-Cover} problem clearly belongs to the class NP (it is sufficient to guess a mapping $f:V(G)\cup \Lambda(G)\to V(H)\cup \Lambda(H)$, checking that it is a covering projection can clearly be done in time polynomial in the size of the input graph $G$). Thus our concern is only to show the NP-hardness of the problem. 
In the first two subsections we will prove it for the case when $H$ is bipartite (and hence does not contain loops nor semi-edges). By the celebrated K\"onig-Hall theorem, such a $k$-regular graph is $k$-edge-colorable. In fact, in the auxiliary constructions presented in Subsection~\ref{subsec:multicover}, we only need the assumption that $H$ is $k$-regular and $k$-edge-colorable. In Subsection~\ref{subsec:reduction}, we revisit the reduction from $k$-edge-colorability of $(k-1)$-uniform $k$-regular hypergraphs to the {\sc $H$-Cover} problem presented in~\cite{n:KPT97} and show how it can be used for bipartite graphs $H$ with multiedges. Finally, in Subsection~\ref{subsec:nonbipartite}, we prove the NP-hardness part of Theorem~\ref{thm:mainresult} for non-bipartite $H$.}

\subsection{Multicovers}\label{subsec:multicover}

The following operation will be an important tool used for building gadgets in our NP-hardness proof. 
\textcolor{red}{In this and the next subsection we only consider graphs without loops or multiple edges, i.e., graphs with only ordinary edges. For such a graph $G$, $\Lambda(G)=E(G)$, and in this sense we also write $E_G(u)=\Lambda_G(u)$ for the set of (ordinary) edges incident with a vertex $u$.}

\begin{definition}[Colored product]\label{def:coloredproduct}
\begin{compactenum}
\item Let $M_1,M_2, \ldots, M_m$ be $m$ perfect matchings, possibly on different  vertex sets. Their product is the graph 
$$\prod_{i=1}^m M_i=(\prod_{i=1}^m V(M_i),\{uv:u_iv_i\in M_i\mbox{ for each }i=1,2,\ldots,m\})$$
where it is assumed that the notation of the vertices of the product is such that $u=(u_1,u_2,\ldots,u_m)$ with $u_i\in V(M_i)$ for all $i=1,2,\ldots,m$. Then $\prod_{i=1}^m M_i$ is a perfect matching as well.
\item Let $G_1,G_2,\ldots,G_m$ be $k$-regular $k$-edge-colorable graphs without loops or semi-edges. For each $i=1,\ldots,m$, let $\phi_i:E(G_i)\longrightarrow\{1,2,\ldots,k\}$ be a proper edge-coloring of $G_i$. The {\em colored product} of $G_i$'s is the graph  $\prod_{i=1}^m G_i$ with vertex set being $\prod_{i=1}^m V(G_i)$ and edge set being the union of $E(\prod_{i=1}^m M^j_i)$, $j=1,2,\ldots,k$, where for each $i$ and $j$, $M^j_i=\phi_i^{-1}(j)$ is the perfect matching in $G_i$ formed by edges colored by color $j$ in the coloring $\phi_i$. If we define $\phi:E(\prod_{i=1}^m G_i)\longrightarrow\{1,2,\ldots,k\}$ by setting $\phi(e)=j$ if and only if $e\in E(\prod_{i=1}^m M^j_i)$, we see that $\phi$ is a proper $k$-edge-coloring of $\prod_{i=1}^m G_i$.
\item Let $G_1,G_2,\ldots,G_m$ and $\phi_i:E(G_i)\longrightarrow\{1,2,\ldots,k\}$ be as in 2) above. For each $i=1,2,\ldots,m$, define the projection $\pi_i$ from the colored product $\prod_{i=1}^m G_i$ to its $i$-th coordinate by setting
$\pi_i(u)=u_i$ and $\pi_i(e)=e_i$ for such an edge $e_i\in E(G_i)$ that satisfies $\phi_i(e_i)=\phi(e)$ and whose end-vertices are $u_i$ and $v_i$, provided the end-vertices of $e$ are $u$ and $v$.
\end{compactenum}
\end{definition}

An example of the colored product of two graphs  is depicted in Figure~\ref{fig:colored_product}. The next lemma follows immediately from Definition~\ref{def:coloredproduct}.

\begin{lemma}\label{lem:coloredproduct}
Let $G_1,G_2,\ldots,G_m$ be $k$-regular graphs without loops or semi-edges, and let for each $i=1,\ldots,m$, $\phi_i:E(G_i)\longrightarrow\{1,2,\ldots,k\}$ be a proper edge-coloring of $G_i$. 
Then each projection $\pi_i$, $i=1,2,\ldots,m$, is a covering projection from $\prod_{i=1}^m G_i$ onto $G_i$. \qed
\end{lemma}

Now we will show the main construction used to build the gadgets for our NP-hardness reduction.

\begin{proposition}
\label{prop:multicoverexistence}
Let $H$ be a connected $k$-regular $k$-edge-colorable graph with no loops or semi-edges.  Then there exists a connected simple $k$-regular $k$-edge-colorable graph $G$ and a vertex $u \in V(G)$ such that 
for every vertex $x\in V(H)$ and for any bijection from $E_G(u)$ onto $E_H(x)$, there exists a covering projection from $G$ to $H$ which extends this bijection and maps $u$ to $x$.
\end{proposition}

\begin{figure}
\centering
{\includegraphics[width=0.9\textwidth]{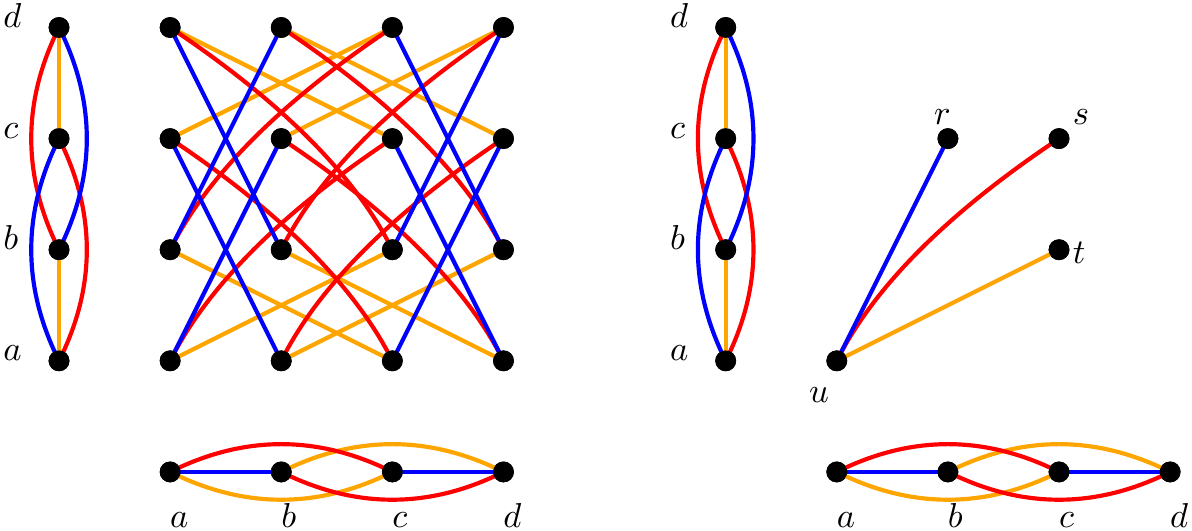}}
\caption{An example of the colored product is on the left. Projections of vertex $u$ and its neighborhood are visualized on the right -- $u,r,s,t$ map in this order onto $a,b,c,c$ in the vertical projection, and onto $a,c,c,b$ in the horizontal one.}
\label{fig:colored_product}
\end{figure}

\begin{proof}
In a way similar to the proof in \cite{n:KPT97} we first construct a colored product of many copies of $H$ that covers $H$ in many ways. 
For every vertex $x_i\in V(H)$, we take $k!$ copies of $H$ with edge colorings obtained by all permutations of colors, and in their colored product denote by  $a_i$ the vertex whose all projections are $x_i$, for each $x_i\in V(H)$. The edges incident with $a_i$ are projected onto $E_H(x_i)$ in all possible ways from this colored product. Then we take $n=|V(H)|$ copies of this product, set $G$ to be their colored product and set $u=(a_1,a_2,\ldots,a_n)$. It follows from Lemma~\ref{lem:coloredproduct} that all projections are covering projections onto $H$, while in the first $k!$ of them, $u$ is projected onto $x_1$, in the second group onto $x_2$, etc., for each $i$ in all possible ways concerning the bijection of $E_G(u)$ onto $E_H(x_i)$. 

 The difference to the approach in \cite{n:KPT97} is that it is not sufficient to require that all bijections of the vertex neighborhoods of $u$ and $x$  can be extended to covering projections, but we must aim at extending bijections of the sets of incident edges. 
 
 The product that we construct may still contain multiple edges. If this is the case, we further take the product with a simple $k$-regular $k$-edge-colorable graph, say $K_{k,k}$. This product is already a simple graph and still possesses all the requested covering projections. It may still be disconnected, though, and we denote by $G$ the component that contains the vitally important vertex $u$ in such a case. 
\qed
\end{proof}

The key building block of our NP-hardness reduction will be the graph $G_u$ obtained from $G$ by splitting vertex $u$ into $k$ pendant vertices of degree 1 (see Figure~\ref{fig:splitting_vertex2G_u}). For each edge $e$ of $G$ incident with $u$, we formally keep this edge with the same name in $G_u$ and denote its pendant vertex of degree 1 by $u_e$. 
(Thus, with this slight abuse of notation, $E_G(u)=\bigcup_{e\in E_G(u)}E_{G_u}(u_e)$.) 
Then we have the following proposition.

\begin{figure}
\centering
\includegraphics[width=0.8\textwidth]{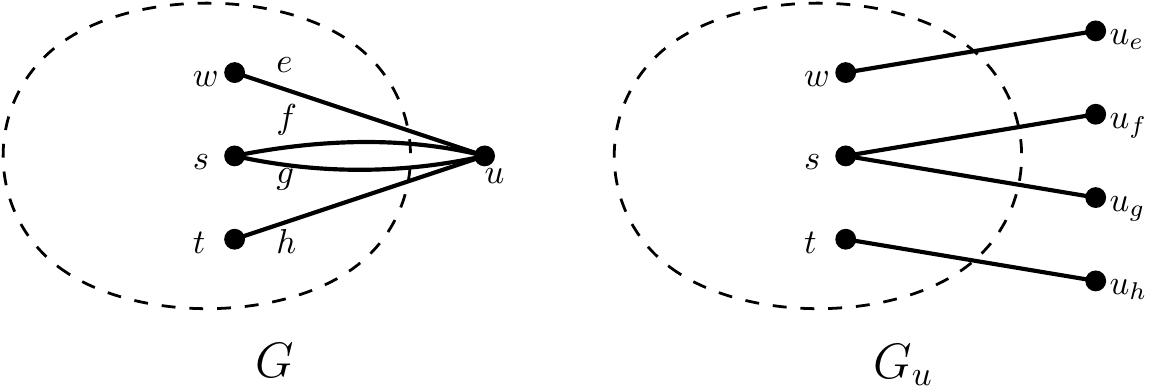}
\caption{An illustration to the construction of $G_u$.}
\label{fig:splitting_vertex2G_u}
\end{figure}

\begin{proposition}
\label{prop:multicoverproperties} 
\textcolor{red}{Let $H$ be a connected $k$-regular $k$-edge-colorable graph with no loops or semi-edges. Then} the graph $G_u$ constructed from the multicover $G$ of $H$ as above satisfies the following:
\begin{compactenum}[(a)]
\item for every vertex $x\in V(H)$ and every bijection $\sigma_x\colon E_G(u) \to E_H(x)$, there exists a partial covering projection of $G_u$ onto $H$ that extends $\sigma_x$ and maps each $u_e, e\in E_G(u)$ to $x$;
\item in every partial covering projection from $G_u$ onto $H$, the pendant vertices $u_e, e\in E_G(u)$ are mapped onto the same vertex of $H$;
\item in every partial covering projection from $G_u$ onto $H$, the pendant edges are mapped onto different edges (incident with the image of the pendant vertices).
\end{compactenum}
\end{proposition}

\begin{proof}
Item (a) follows directly from Proposition~\ref{prop:multicoverexistence}. 
To prove (b) and (c), we first show that in any partial covering projection $f$ from $G_u$ onto $H$, the pendant edges are mapped onto edges of all colors.   

Denote by $E^*$ the pendant edges of $G_u$, and set $V^*=V(G)\setminus \{u\}=V(G_u)\setminus\bigcup_{e\in E_G(u)}\{u_e\}$. 
Observe first that since $H$ has a perfect matching, the number of its vertices is even, and since $G$ covers $H$, so is the number of vertices of $G$. Hence $|V^*|\equiv 1 \mod 2$. 

Suppose $f:G_u\longrightarrow H$ is a partial covering projection. Fix a proper $k$-coloring $\phi$ of the edges of $H$ and define a $k$-coloring $\widetilde{\phi}$ of the edges of $G_u$ by setting $\widetilde{\phi}(e)=\phi(f(e))$. Since $f$ is a partial covering projection, $\widetilde{\phi}$ is a proper edge-coloring of $G_u$. If some color is missing on all of the edges from $E^*$, edges of this color would form a perfect matching in $G_u[V^*]$ and $|V^*|$ would be even, a contradiction. Hence every color appears on exactly one edge of $E^*$. 

For every $a\in V(H)$, denote by $h_a$ the number of vertices of $V^*$ that are mapped onto $a$ by $f$. Consider an edge $e'$ connecting vertices $a$ and $b$ of $H$. If $e'=f(e)$ for some edge $e\in E^*$, we have $f(u_e)=a$ and $f(w_e)=b$, or vice versa. Every vertex in $f^{-1}(a)\cap V^*$ is adjacent to exactly one vertex in $f^{-1}(b)\cap V^*$ via an edge of color $\phi(e')$, whilst every vertex of  $f^{-1}(b)\cap V^*$ except for $w_e$ is adjacent to exactly one vertex in  $f^{-1}(a)\cap V^*$ via an edge of the same color. Hence $h_b=h_a+1$. Orient the edge $e'$ from $a$ to $b$ in such a case. If, on the other hand, the edge $e'$ is not the image of any edge from $E^*$, the edges of color $\phi(e')$ form a matching between the vertices of $f^{-1}(a)\cap V^*$ and the vertices of $f^{-1}(b)\cap V^*$, and $h_a=h_b$. Leave the edge $e'$ undirected in such a case.

After processing all edges of $H$ in this way, we have constructed a mixed graph $\overrightarrow{H}$ which has exactly one edge of each color directed. From the meaning of the orientations and non-orientations of edges of $\overrightarrow{H}$ for the values of $h_a, a\in V(H)$, it follows that the vertex set of $H$ falls into levels, say $L_r, L_{r+1},\ldots,L_s$ such that undirected edges live inside the levels, while directed edges connect vertices of consecutive levels and are directed from $L_i$ to $L_{i+1}$ for suitable $i$. Every two consecutive levels are connected by at least one directed edge in this way. Since $H$ is connected, every vertex $a\in L_i$ will satisfy $h_a=i$ (the indices $r$ and $s$ are chosen so that $r$ is the smallest value of $h_a$ and $s$ is the largest one). Directed edges connecting two consecutive levels form a cut in $H$, and since edges of each color form a perfect matching, the parities of the numbers of edges of each color in this cut are the same. Since the cut contains at least one edge, but at most one edge of each color (exactly one edge of each color is directed), it follows that the cut contains all $k$ directed edges and that $H$ has only two levels, i.e., $s=r+1$. Thus $H$ has $|L_r|$ vertices $a$ with $h_a=r$ and $|L_{r+1}|=|V(H)|-|L_r|$ vertices $a$ with $h_a=r+1$. It follows that 
$$|V^*|=r|L_r|+(r+1)(|V(H)|-|L_r|)=(r+1)|V(H)|-|L_r|.$$
We know that $G$ covers $H$, and so $|V(G)|=\ell|V(H)|$ for some $\ell$. Thus $|V^*|=\ell|V(H)|-1$. Thus we obtain 
$$(r+1)|V(H)|-|L_r|=\ell|V(H)|-1,$$
which implies
$$(r+1-\ell)|V(H)|=|L_r|-1,$$
and since $1\le |L_r|\le |V(H)|-1$, the only possible way for $|L_r|-1$ to be divisible by $|V(H)|$ is $|L_r|=1$. But this implies that all directed edges of $\overrightarrow{H}$ start in the same vertex, say $z$, and from the construction of $\overrightarrow{H}$ it follows that $f(u_e)=z$ for all $e\in E^*$. This proves (b).

Now (c) follows from the two observations above. The $k$ pendant edges of $E^*$ have mutually distinct colors in $\widetilde{\phi}$, and thus they must be mapped to distinct edges of $E_H(z)$ by $f$.
\qed
\end{proof}

\subsection{Reduction from hypergraph coloring}\label{subsec:reduction}

\begin{figure}[t]
\centering
\includegraphics[width=0.9\textwidth]{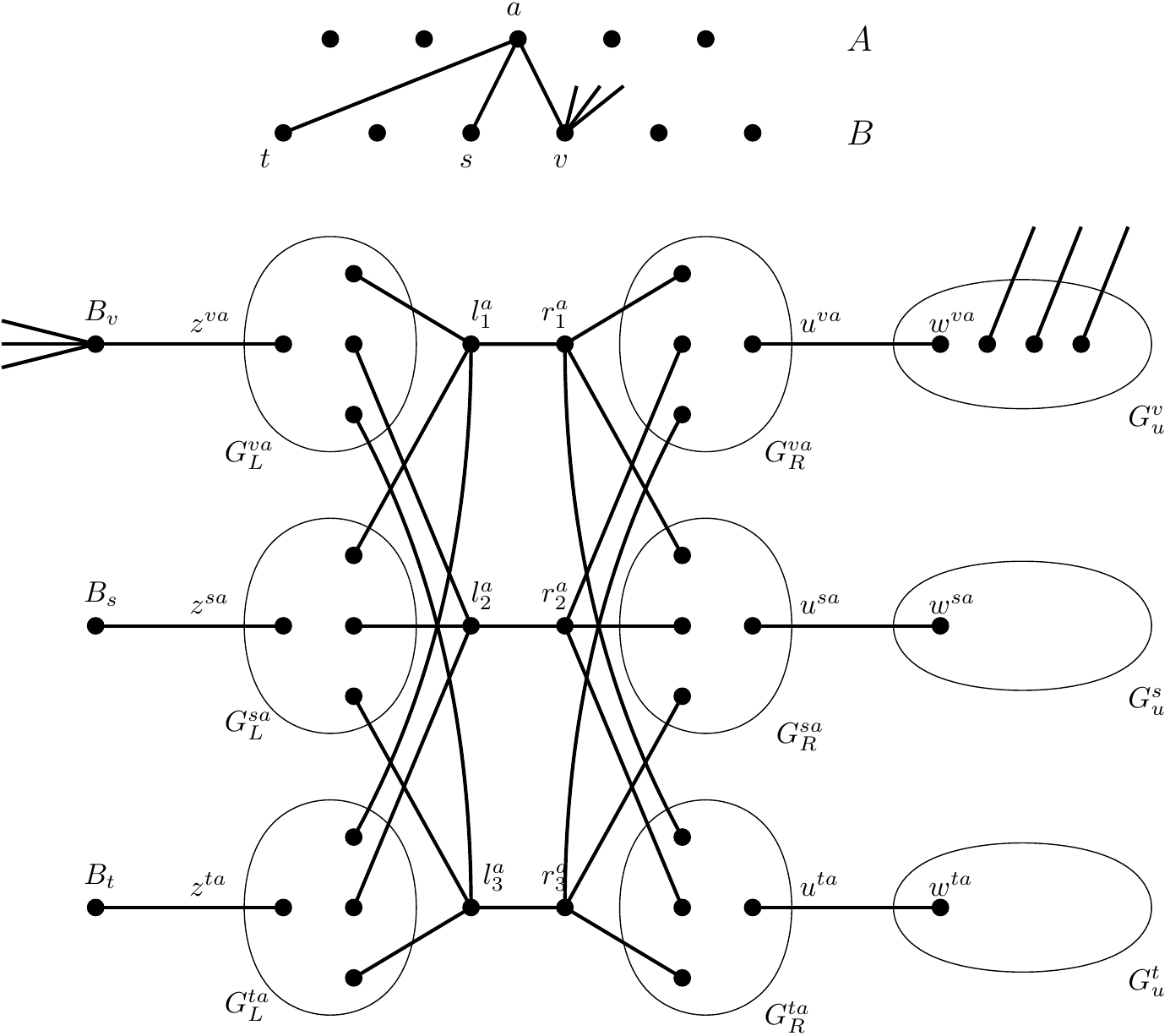}
\caption{An illustration to the construction of $G_K$ for $k=4$.}
\label{fig:construction_G_K}
\end{figure}

\textcolor{red}{
In this subsection, we prove Theorem~\ref{thm:mainresult} for the case of bipartite $H$. Note that a bipartite graph has no loops and no semi-edges, and that it is $k$-edge-colorable if it is $k$-regular. Note also that a semi-simple vertex in a graph with no semi-edges is simple.} 

\textcolor{red}{
\begin{proposition}\label{prop:bipartiteproof}
Let $k\ge 3$ and let $H$ be a connected $k$-regular bipartite graph with a simple vertex. Then {\sc List-$H$-Cover} is NP-hard for simple input graphs.
\end{proposition}}

\begin{proof}
The NP-hardness reduction is exactly the same as in \cite{n:KPT97}, but the proof for the case when multiple edges are allowed in $H$ needs some extra analysis. Hence we need to describe the reduction in full detail in here. The reduction is from $k$-edge-colorability of $(k-1)$-uniform $k$-regular hypergraphs. In the wording of the incidence graph of the hypergraph, suppose we are given a simple bi-regular bipartite graph $K=(A\cup B, E)$ such that all vertices in $A$ (which represent the edges of the hypergraph)  have degree $k-1$ and all vertices in $B$ (which represent the vertices of the hypergraph) have degree $k$. The question is if the vertices of $A$ can be colored by $k$ colors so that the neighborhood of each vertex from $B$ is rainbow colored (i.e., each vertex from $B$ sees all $k$ colors on its neighbors, each color exactly once). This problem is NP-complete for every fixed $k\ge 3$~\cite{n:KPT97}. 

Given such a graph $K$, we build an input graph $G_K$ by local replacements. Recall that we are working with a $k$-edge-colorable $k$-regular graph $H$ with a simple vertex, say $x$. And we are guaranteed the existence of a graph $G_u$  which satisfies the properties stated in Proposition~\ref{prop:multicoverproperties}. The size of $G_u$ is constant with respect to the size of the input graph $K$. And this $G_u$ will be a key building block in our construction.

First, every vertex $v\in B$ will be replaced by a copy of the so-called vertex gadget, which is a disjoint union of a copy $G_u^v$ of $G_u$ and a single vertex $B_v$. For every neighbor $a\in A$ of $v$, one of the pendant vertices of $G_u^v$ will be denoted by $u^{va}$, and its neighbor within $G_u^v$ will be denoted by $w^{va}$. 

The hyperedge gadgets used to replace the vertices of $A$ are more complicated. This gadget consists of $2(k-1)$ copies of $G_u$ linked together in the following way. Let $a\in A$. We take $2(k-1)$ vertices $\ell^a_i, r^a_i, i=1,2,\ldots,k-1$, and for every neighbor $v$ of $a$, we take two copies $G^{va}_L$ and $G^{va}_R$ of $G_u$. The pendant vertices of $G^{va}_L$ will be unified with $B_v$ and $\ell^a_1, \ell^a_2,\ldots,\ell^a_{k-1}$, while the pendant vertices of $G^{va}_R$ will be unified with $w^{va}$ and   $r^a_1, r^a_2,\ldots,r^a_{k-1}$. The neighbor of $B_v$ in $G^{va}_L$ will be denoted by $z^{va}$. Lastly, the matching $\ell^{a}_ir^a_i, i=1,2,\ldots,k-1$ is added. This completes the construction of $G_K$. See Figure~\ref{fig:construction_G_K} for an illustrative example.

The resulting graph $G_K$ is $k$-regular. To make it an instance of the {\sc List-$H$-Cover} problem, we prescribe that the vertices $B_v, v\in B$ and $\ell^a_i, a\in A, i=1,2,\ldots,k-1$ are all mapped onto $x$ (this means, that for these vertices, their lists of admissible target vertices are one-element and all the same, while for the remaining vertices, their lists are full, as well as for all the edges).

The fact that $x$ is simple implies the following observation: For every partial covering projection from $G_u$ to $H$ which maps all the pendant vertices onto $x$, their neighbors in $G_u$ are mapped onto distinct vertices of $H$ (this immediately follows from property (c) of Proposition~\ref{prop:multicoverproperties}). Similarly, if any vertex of $G_K$ is mapped onto $x$ by a covering projection, then its neighbors are mapped onto distinct vertices of $H$ (the neighborhood $N_H(x)$ of $x$ in $H$). We will exploit these observations in the following argumentation.

Suppose $f:G_K\longrightarrow H$ is a covering projection such that all vertices $B_v, v\in B$ and $\ell^a_i, a\in A, i=1,2,\ldots,k-1$ are mapped onto $x$. Consider an $a\in A$, and let $f(r_1^a)=y$, whence $y\in V(H)$ is a neighbor of $x$ in $H$. Property c) applied to any $G^{va}_R$, for $v$ being a neighbor of $a$ in $K$, implies that $f(r^a_i)=f(w^{va})=y$ for all $i=2,\ldots,k-1$. Since each $\ell^a_i$ has a neighbor $r^a_i$ mapped onto $y$, none of their neighbors in $G^{va}_L$ is mapped onto $y$. Property (d) then implies that $f(z^{va})=y$. Define a coloring $\phi$ of $A$ by colors $N_H(x)$ as $\phi(a)=f(r^a_1)$. Consider a vertex $v\in B$. The neighbors of $B_v$ in $G_K$ are $z^{va}, a\in N_K(v)$. Since $f(B_v)=x$ and $x$ is simple, the vertices  $z^{va}, a\in N_K(v)$ are mapped onto different neighbors of $x$ by $f$, and hence the colors $\phi(a), a\in N_K(v)$ are all distinct. Thus $\phi$ is a $k$-coloring of $A$ of the required property.

Suppose for the opposite direction that $A$ allows a $k$-coloring $\phi$ such that each vertex $v\in B$ sees all $k$ colors on its neighbors, and identify the colors with the names of the neighbors of $x$ in $H$. Furthermore, set $f$ in the following way:
\begin{align*}
f(B_v) &= \ell^a_i=x\mbox{ for all $v\in B, a\in A, i=1,2,\ldots,k-1$ (as required by the lists)}, \\ 
f(u^{va}) &= x\mbox{ for all $v\in B$ and $a\in N_K(v)$}, \mbox{and}\\
f(r^a_i) &= f(w^{va})=f(z^{va})=\phi(a)\mbox{ for all $a\in A$ and $v\in N_K(a)$}.
\end{align*}
Finally, define $f$ on the edges incident to $\ell^a_i$ ($r^a_i$, respectively) so that for every $i$, these edges are mapped onto different edges incident to $x$ (to $\phi(a)$, respectively), and, on the other hand, for every $v\in N_K(a)$, the pendant edges of $G^{va}_L$ (of $G^{va}_R$, respectively) are mapped onto distinct edges incident to $x$ (to $\phi(a)$, respectively).

The property (a) of Proposition~\ref{prop:multicoverproperties} implies that this mapping can be extended to partial covering projections within each copy of $G_u$ used in the construction of $G_K$. To see that they altogether provide a covering projection from $G_K$ to $H$, note that for each $v\in B$, the edges incident with the vertex $B_v$ are mapped onto different edges because their other endpoints are $\phi(a), a\in N_K(v)$, and hence all different by the assumption on the coloring $\phi$, and also each copy $G^v_u$ has its pendant edges mapped onto different edges incident to $x$, since the pendant vertices $u^{va}, a\in N_K(v)$ are all mapped onto $x$ and their neighbors $w^{va}$ in $G^v_u$ are mapped onto distinct vertices $\phi(a), a\in N_K(v)$. This concludes the proof.   
\qed
\end{proof}

\subsection{The non-bipartite case}\label{subsec:nonbipartite}

\begin{proof}[of Theorem~\ref{thm:mainresult}]\textcolor{red}{
If $H$ is bipartite, the claim follows from Proposition~\ref{prop:bipartiteproof}.}
Suppose the graph $H$ is not bipartite (this includes the case when $H$ contains loops and/or semi-edges). Consider $H'=H\times K_2$, where $\times$ denotes the categorical product, see an example in Fig.~\ref{fig:sau_to_rings}. This $H'$ may still contain multiple edges (the product of a multiple ordinary edge with $K_2$ is again a multiple edge, but also the product of a loop with $K_2$ is a double ordinary edge, and the product of a multiple semi-edge with $K_2$ results in a multiple ordinary edge as well), but it is bipartite (and thus has neither semi-edges nor loops) and therefore is $k$-edge-colorable. In the product with $K_2$, every semi-simple vertex of $H$ results in two simple vertices of $H'$. Hence, by the result of the preceding subsection, {\sc List-$H'$-Cover} is NP-complete. 

It is proved in~\cite{n:FK08} that for simple graphs, $G$ covers $H\times K_2$ if and only if $G$ is bipartite and covers $H$. This proof readily extends to graphs $H$ that allow loops, semi-edges and multiple edges. The proof for the list version of the problem may get more complicated in general. However, the list version that we have proven NP-complete in the preceding subsection is very special: the lists of all edges are full, and so are the lists of all the vertices except for those which are prescribed to be mapped onto the same simple vertex, say $x'$. If we take such an instance of {\sc List-$H'$-Cover}, this $x'$ is a copy of a semi-simple vertex $x\in V(H)$, and all vertices of the input graph $G$ that are prescribed to be mapped onto $x'$ are from the same class of its bipartition. We just prescribe them to be mapped onto $x$ as an instance of {\sc List-$H$-Cover}. It is easy to see that this mapping can be extended to a covering projection to $H$ if and only if $G$ allows a covering projection to $H'$ in which all these prescribed vertices are mapped onto $x'$. This concludes the proof.
\qed
\end{proof}

\section{Sausages and rings}\label{sec:sausages}

In this section we consider two special classes of cubic graphs. These graphs play a special role in the classification in Theorem~\ref{thm:cubicdichotomy}. The {\em $k$-ring} (where $k \geq 2$) is the cubic graph obtained from the cycle of length $2k$ by doubling every second edge. We call a {\em $k$-sausage} every cubic graph that is obtained from a path on $k$ vertices by doubling every other edge and adding loops or semi-edges to the end-vertices of the path to make the graph 3-regular. Note that while for every $k$, the $k$-ring is defined uniquely, there are several types of $k$-sausages, as depicted in Figure~\ref{fig:sausages_and_rings}.

 \begin{figure}
\centering
\includegraphics[width=0.9\textwidth]{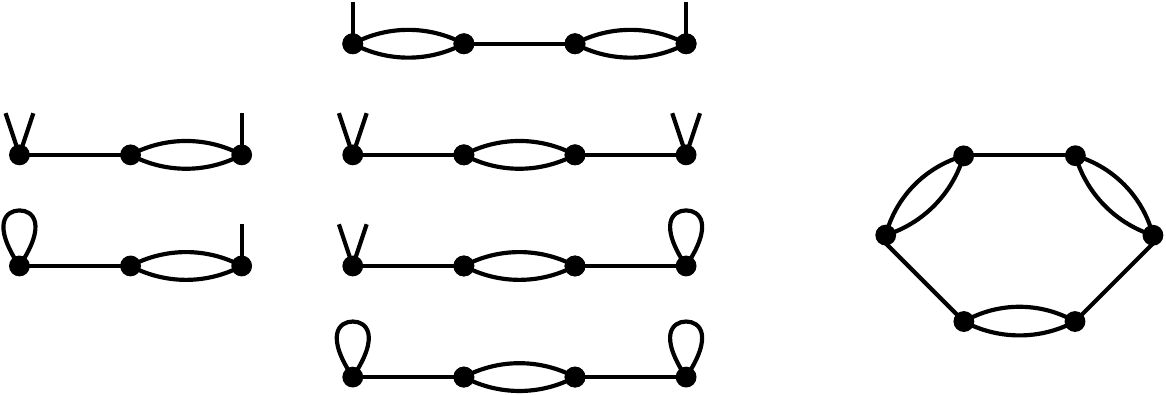}
\caption{The two non-isomorphic 3-sausages (left), the four non-isomorphic 4-sausages (middle), and the 3-ring (right).}
\label{fig:sausages_and_rings}
\end{figure}

\begin{proposition}\label{prop:sau_and_rings}
Let $k\ge 2$ and let $S_k$ be a $k$-sausage. Then $S_k\times K_2$ is isomorphic to the $k$-ring.
\end{proposition}
\begin{proof}
The product $H\times K_2$ is a bipartite graph with no loops or semi-edges, in which every ordinary edge in $H$ gives rise to a pair of ordinary edges of the same multiplicity. \textcolor{red}{ A loop, as well as} a pair of semi-edges incident to the same vertex of $H$, gives rise to a double ordinary edge. A single semi-edge in $H$ gives rise to a simple ordinary edge in $H\times K_2$. Thus $S_k\times K_2$ has a cyclic structure and the number of double edges is equal to the number of vertices of $S_k$, see Figure~\ref{fig:sau_to_rings}. \qed   
\end{proof}

\begin{figure}
\centering
\includegraphics[width=0.4\textwidth]{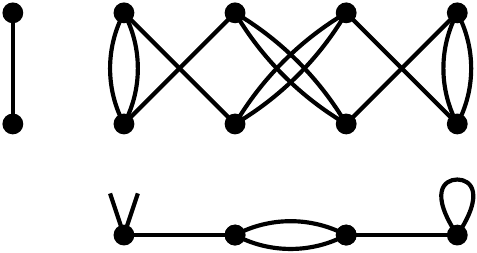}
\caption{The product of a $4$-sausage with $K_2$ is isomorphic to the $4$-ring.}
\label{fig:sau_to_rings}
\end{figure}

In the following three theorems we show that the \textsc{List-$k$-ring-Cover} problem is NP-complete for simple graphs on the input for every $k \geq 3$. (In other words, for all rings, except for the $2$-ring, for which the  NP-completeness follows even without lists from \textcolor{red}{Theorem~18 of~\cite{n:BFHJK-MFCS}}.) \textcolor{blue}{Note also that in case of $k$ having an odd-prime divisor (Theorem~\ref{thm:rings}) and $k=4$ (Theorem~\ref{thm:4ring}), we prove the NP-hardness in a stronger form, i.e., without lists.}

\begin{theorem}
\label{thm:rings}
The {\sc $k$-ring-Cover} problem is NP-complete for simple input graphs for every $k=2^{\alpha}(2\beta+3)$ such that $\alpha$ and $\beta$ are non-negative integers.
\end{theorem}

\begin{proof}
We reduce from {\sc $C_{(2\beta+3)}$-Hom} for $\beta$ being a non-negative integer which is known to be NP-complete by the dichotomy theorem of Hell and Nešetřil~\cite{hn} (the {\sc $H$-Hom} problem asks for the existence of an unconstrained homomorphism from the input graph to the parameter graph $H$). We call $k$-ring occasionally $H$. Furthermore, the vertices of $H$, i.e.\ $k$-ring, will be consecutively denoted by $1,1',\ldots,k,k'$ with precisely edges $jj'$ being double for $j \in \{1,\ldots,k\}$.

The reader is advised to consult Figure~\ref{fig:5ring} to better understand the gadgets and the reduction.

Let us describe the vertex gadget. Suppose we have a vertex $v$ of degree $\mathrm{deg}(v)$. Then let us take two disjoint copies of $C_l$ where $l = 2k \times \mathrm{deg}(v)$. Let us denote them $C^V_1$ and $C^V_2$ and their vertices $v_{1,1},\ldots,v_{1,l}$. The construction of the gadget proceeds in the following way:
\begin{itemize}
  \item Add edges $v_{1,2i-1}v_{2,2i}$ for every $i \in \{1,\ldots, l/2\}$.
  \item Add edges $v_{2,2i-1}v_{1,2i}$ for every $i \in \{1,\ldots, l/2\}$.
  \item Delete edge $v_{2,2kj-1}v_{2,2kj}$ for every $j \in \{1,\ldots, \mathrm{deg}(v)\}$ and to each of the endpoints of the deleted edge, add a pendant vertex.
\end{itemize}

We call these pendant vertices leaves of the vertex gadget and we speak about \emph{pairs of leaves} when we refer to the two pendant vertices created after the deletion of the same edge. Further, since the gadget is bipartite, we can say that leaves are either black or white depending on the part of the bipartition they belong to. Observe that every pair of leaves in the above sense has one black and one white vertex. 

Before we describe the edge gadget, let us introduce \emph{enforcing gadget} which is simply the same as the vertex gadget for vertex of degree 1, i.e. it is created as was described in the preceding with $l = 2k$.

In the following, we number the vertices of every cycle consecutively starting with 1. We can thus speak about e.g. the $i$-th even vertex. Also, we automatically take the indices of cycles modulo $2k$.

For the actual edge gadget, take $k$ disjoint copies of cycles $C_{2k}$. Let us denote these cycles $C^E_1,\ldots,C^E_k$. We now insert enforcing gadgets in between the cycles as follows.
For all $j$ being odd and $j < k$ and for all even vertices of $C^E_j$, we identify one of the leaves of the enforcing gadget with the $i$-th even vertex of $C^E_j$, and we identify the other leaf of the enforcing gadget with the $i$-th even vertex of $C^E_{j+1}$. 

The similar connection is done in case of $j$ even and $j < k$, except that the $i$-th odd vertex of $C^E_j$ is connected by a copy of enforcing gadget to the $i$-th odd vertex of~$C^E_{j+1}$. 

Now, except for $C^E_1$ and $C^E_k$, all vertices are of degree 3. For every vertex of degree 2 in $C^E_1$ except for the first and the $(1+2^{\alpha+1})$-th vertex, let us say the $i$-th one, we place the enforcing gadget between the $i$-th vertex of $C^E_1$ and the $(i+k)$-th vertex of $C^E_k$ again by identifying each of the leaves with one of the mentioned vertices. This completes the construction of the edge gadget. The only vertices of degree 2 are now the first and the $(1+2^{\alpha+1})$-th vertex in $C^E_1$ and the $(1+k)$-th and the $(1+2^{\alpha +1}+k)$-th vertex in $C^E_k$.

Let us have an instance $G$ of {\sc $C_{(2\beta+3)}$-Hom}. We shall construct a new graph $G'$. For each vertex in $G$, we take a copy of the vertex gadget of corresponding size and insert it into $G'$. For each edge $uv$ in $G$, we obtain a new copy of the edge gadget. We connect it with the vertex gadget corresponding to $u$ as follows. We take one of the so-far unused pair of leaves coming from the vertex gadget corresponding to $u$ and identify the black leaf of the pair with the first vertex in $C^E_1$ of the edge gadget and the white leaf of the pair with the $(1+k)$-th vertex of $C^E_k$. For the vertex gadget of $v$, we again take one of the so-far unused pair of leaves coming from the vertex gadget of $v$ and identify the black leaf of the pair with the $(1+2^{\alpha+1})$-th vertex in $C^E_1$ of the edge gadget and the white leaf of the pair with the $(1+2^{\alpha+1}+k)$-th vertex of $C^E_k$.

We shall now describe possible images of vertex and edge gadget under a covering projection to $k$-ring.

We claim that under every covering projection to $H$, all black leaves of a given vertex gadget will be mapped to the same vertex of $k$-ring and white vertices to its prime version (or vice versa). The crucial observation is that $v_{1,1},v_{2,2},v_{2,1},v_{1,2}$ form a 4-cycle in the vertex gadget. By a simple analysis then, $v_{1,1}$ and $v_{2,1}$ must be mapped to some $\ell$ of $H$ and $v_{1,2}$ and $v_{2,2}$ to $\ell'$ (or vice versa, but let us without loss of generality assume the first possibility). Furthermore this enforces the images of vertices $v_{1,3},v_{2,4},v_{2,3},v_{1,4}$ as well (and they form again a 4-cycle). A repeated use of this propagation ensures that images of $v_{2,2k-1}v_{2,2k}$ are $(\ell-1)$ and $(\ell - 1)'$, respectively (and possibly modulo $2k$, which will be assumed from now on). By the construction, the pendant vertices have then images $(\ell - 1)'$ and $(\ell - 1)$. The argument then can be repeated further and further, until we arrive on the conclusion that all black leaves have inevitably the same image $(\ell - 1)'$ and the white leaves $(\ell - 1)$.

Specially, for the enforcing gadget, we get that its leaves must be mapped to the different endpoints of a specific double edge in $k$-ring. In other words, whenever we have a vertex which is being identified with one of the leaves of the enforcing gadget, then given its image $i$ under a covering projection to the $k$-ring, the other vertex identified with the other leaf of the enforcing gadget has to be mapped to $i'$ in $k$-ring or vice versa.

We claim that under every covering projection, the edge gadget will be mapped to $k$-ring in the following way. Without loss of generality, the first vertex of $C^E_1$, let us call it $a$, will be mapped to $1$. Clearly, as the edge gadget is connected here to a vertex gadget through $a$, one of the neighbors of $a$ in the edge gadget has to be mapped to $1'$ and the other to $k'$, or vice versa. In both cases these neighbors are connected through enforcing gadget to $C^E_2$ and thus this enforces not only the images of vertices at distance 2 from $a$ on $C^E_1$ but also the images of the vertices at distance two from $a$ on $C^E_2$. Proceeding inductively, we arrive on the conclusion that the vertices of $C^E_1$ are either consecutively $1,1',2,2',\ldots,k,k'$ or in the counter-clockwise fashion $1,k',k,\ldots,1'$. Let us, again without loss of generality, describe what follows in the first situation. The other one is in fact just a mirrored situation and the argumentation is thus almost the same. In the first situation, the images of $C^E_2$ are shifted clockwise by one, so the images of vertices of $C^E_2$ are consecutively $k',1,1',\ldots,k$. This propagates the shifting further to $C^E_3$ and inductively up to $C^E_k$, also thank to the enforcing gadgets between the layers.

Up to now, everything was enforced, all vertices of the edge gadget are mapped. However, it remains to check two things.
\begin{enumerate}
  \item Whether all the vertices identified with leaves of vertex gadgets are mapped in the right way.
  \item Whether the edges between $C^E_1$ and $C^E_k$ and their endpoints do have the right images.
\end{enumerate}
Regarding (1), we assumed $a$ is mapped to $1$. Thus the other, white leaf from the pair containing $a$ has to be mapped to $1'$. This is indeed all right considering the shifting of images between the cycles. The same can be argued for the pair of leaves corresponding to the second vertex gadget attached. The argumentation here is based on the fact that the black leaf of the other vertex gadget is identified with vertex at distance $2^{\alpha+1}$ on $C^E_1$.

Case (2) depends on an analogous argument, again based on shifting of the images of $C^E_1$ versus the images of $C^E_k$.

Finally, let us observe that $G'$ is again bipartite and thus it remains valid to speak about black and white leaves of vertex gadgets (or about vertices corresponding to such leaves). Furthermore, we can say that $G'$ has black and white vertices.

\begin{figure}
\centering
{\includegraphics[width=\textwidth]{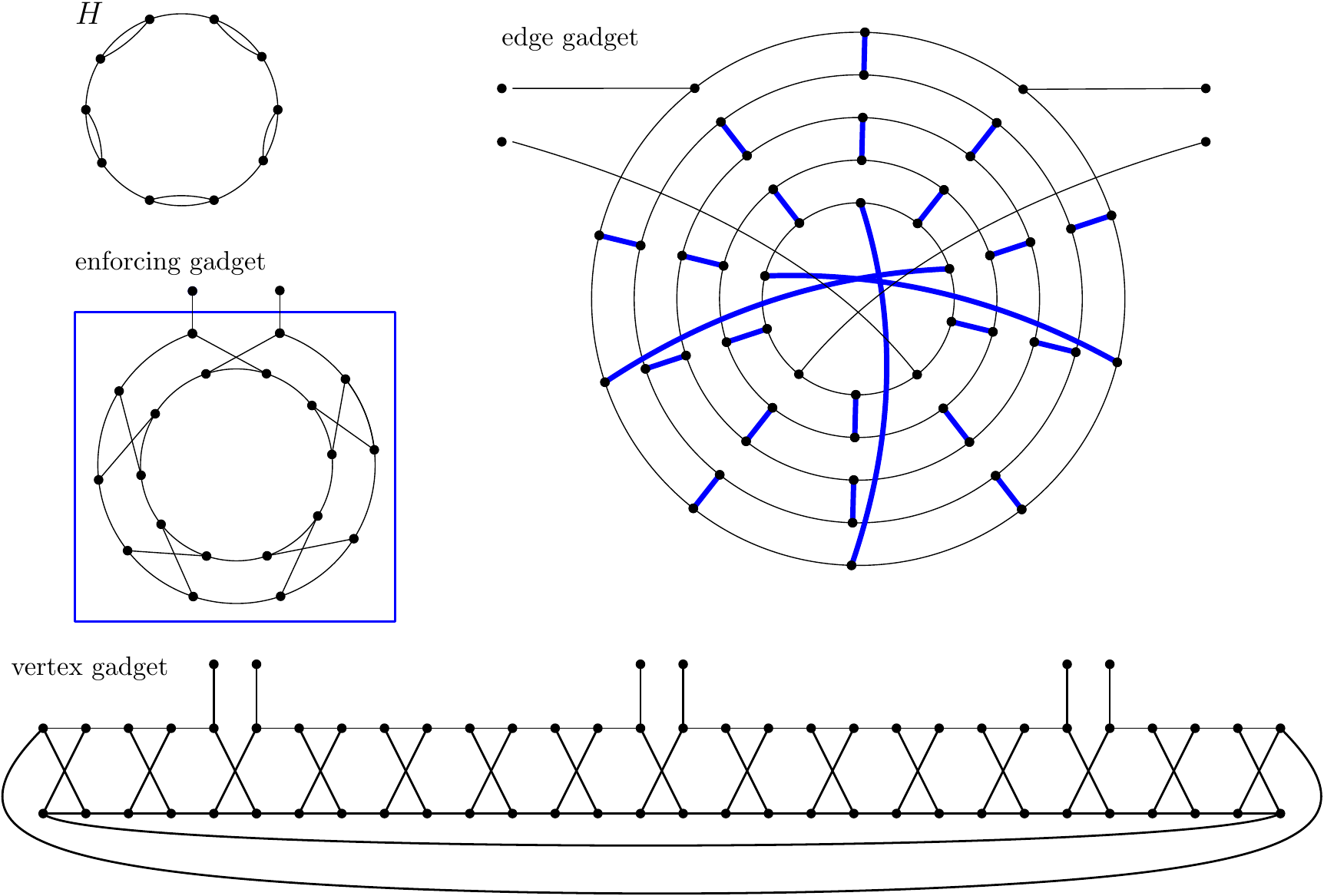}}
\caption{An example of the construction for $H$ being 5-ring.}
\label{fig:5ring}
\end{figure}

Suppose that $G'$ covers $H$. Clearly, as was described and without loss of generality, all vertices corresponding originally to the black leaves of all vertex gadgets in $G$ must be mapped to the black vertices of $H$ and white leaves to the white vertices of $H$. We can further assume without loss of generality that black vertices of vertex gadgets map to vertices $1+(j-1)(2^{\alpha+1})$ where $j \in \{1,\ldots,2\beta - 3 \}$ of $H$.

Edge gadgets enforce that whenever there is an edge in $G$, then the images of the respective black leaves of vertex gadgets connected to it are exactly at distance $2^{\alpha + 1}$ from each other. This implies that there is a homomorphism of $G$ to $C_{(2\beta+3)}$ and we can say that the image of $v$ is the $j$-th vertex of $C_{(2\beta+3)}$ if the black leaves of the vertex gadget of $v$ map to $1+(j-1)(2^{\alpha+1})$.

For the other direction, suppose there exists a homomorphism of $G$ to $C_{(2\beta+3)}$. If a vertex $v$ of $G$ is mapped to the $j$-th vertex of $C_{(2\beta+3)}$, we map the vertices corresponding to the black leaves of the vertex gadget of $v$ to the vertex $1+(j-1)(2^{\alpha+1})$ of $H$. We already know that this ensures that there is only one possible mapping of vertices of the vertex gadget of $v$ to $H$. The same can be done and said for all the other vertex gadgets of $G'$. From the analysis of possible mappings of the edge gadgets, we know than we can complete the mapping now so that the result is a covering projection of $G'$ to $H$.

We showed that there is a homomorphism of $G$ to $C_{(2\beta+3)}$ if and only if $G'$ covers $H$. This completes the reduction and thus the proof of the theorem.
\qed
\end{proof}

\begin{theorem}
\label{thm:ringspowersoftwo}
The {\sc List-$k$-ring-Cover} problem is NP-complete for simple input graphs for every $k=2^{\alpha}$ such that $\alpha\ge 3$ is an integer.
\end{theorem}

\begin{proof}
By a seminal result of Feder et al.~\cite{feder1999list}, {\sc List-$H'$-Hom} is NP-complete if $H'$ is a so called bi-arc graph. By Corollary 3.1. therein, it follows that all cycles of size at least five are bi-arc graphs. Thus, we shall reduce from {\sc List-$C_k$-Hom}.

The construction of $(G',L')$ for a given $(G,L)$ being an instance of {\sc List-$C_k$-Hom} is almost the same as in Theorem~\ref{thm:rings}. We shall only describe differences here. Suppose that $C_k$ has its vertices named consecutively $1,2,\ldots,k$. Again as before, the $k$-ring will have its vertices denoted consecutively by $1,1',\ldots,k,k'$ with precisely edges $jj'$ being double for $j \in \{1,\ldots,k\}$.  
\begin{itemize}
  \item Suppose we have a vertex $v$ in $G$. Then for every $\ell \in L(v)$, we add $\ell$ to the lists of all black leaves of its vertex gadget and $\ell'$ to all lists of all white leaves. All of the other vertices in $G'$ will have full lists, i.e.\ all vertices of $k$-ring. Also the edges will have full lists.
  \item Vertex gadgets will be connected to the respective edge gadget in a slightly different way. For an edge $uv$ of $G$, we shall identify the black leaf of a pair of a vertex gadget for $u$ with first vertex of $C^E_1$ and the white leaf of the same pair with the $(1+k)$-th vertex of $C^E_k$ of the edge gadget for $uv$. We shall identify the black leaf of a pair of a vertex gadget for $v$  with the third vertex of $C^E_1$ and the white leaf of the same pair with the $(3+k)$-th vertex of $C^E_k$ of the edge gadget for $uv$.
\end{itemize} 

Clearly, under any possible list covering projection, all black leaves of a given vertex gadget choose simultaneously one vertex $\ell$ from their lists and subsequently $\ell'$ for all its white leaves (or vice versa).

Now the same analysis can be done as in Theorem~\ref{thm:rings} to show that there exists a list homomorphism of $(G,L)$ to $C_k$ if and only if there exists a list covering projection of $(G',L')$ to $k$-ring. This completes the proof.
\qed
\end{proof}

\begin{theorem} \label{thm:4ring}
The \textsc{4-ring-Cover} problem is NP-complete for simple input graphs.
\end{theorem}

\begin{proof}
The case of $4$-rings needs a special ad hoc construction. However, the ideas are very similar to the previous ones. In order to not repeat ourselves, we shall describe a sketch of the reduction with the help of Figure~\ref{fig:4ring} (referred to as figure in the rest of the proof). This time, we shall reduce from the problem of \textsc{4-Coloring}.

\begin{figure}
\centering
{\includegraphics[width=\textwidth]{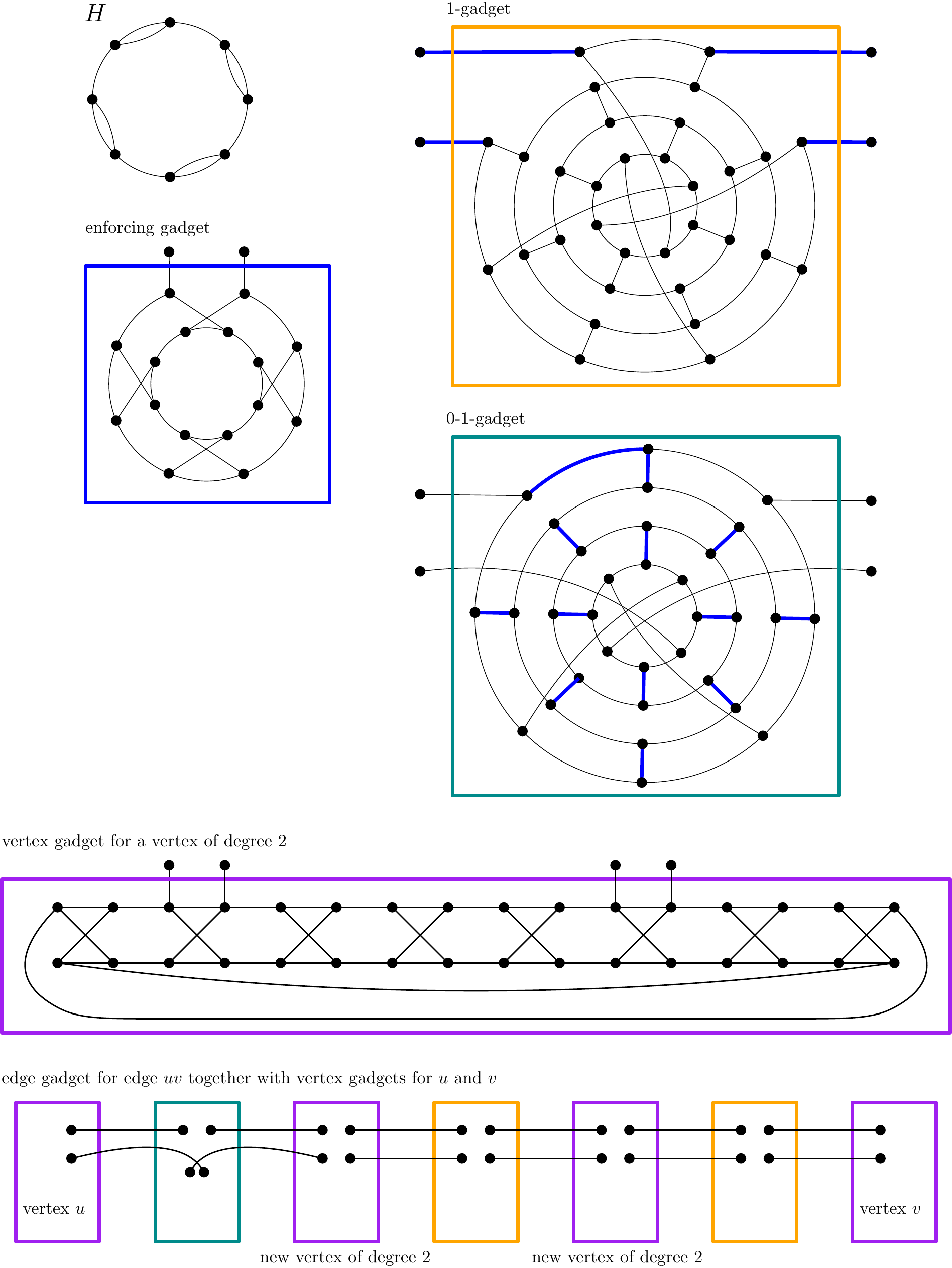}}
\caption{An example of the construction for $H$ being 4-ring.}
\label{fig:4ring}
\end{figure}

Suppose we have a graph $G$ as an instance of \textsc{4-Coloring}. We will construct $G'$ as an instance of \textsc{4-ring-Cover} in the following way. First, we shall insert into $G'$ vertex gadgets, the same ones as in the previous reductions on rings as one can see in the figure.

We introduce two new gadgets that will form our edge gadgets here: \emph{1-gadget} and \emph{0-1-gadget}. We shall also use the enforcing gadget we introduced in Theorem~\ref{thm:rings}. The figure shows examples of explicit constructions of such gadgets. For simplicity, let us denote the vertices outside the respective boxes of gadgets (which are technically in the gadgets) as \emph{terminal vertices} or \emph{terminals}.

Let us now discuss the effect of the newly introduced gadgets on its terminal vertices.
\begin{itemize}
  \item 1-gadget: If the left terminal vertices have $i$ and $i'$ as images, the right terminal vertices will have $i+1$ and $(i+1)'$ as their images (where numbers are taken modulo 4).
  \item 0-1 gadget: This gadget has a similar effect. If the left terminals have $i$ and $i'$, the right terminals will have either $i$ and $i'$, or $i+1$ and $(i+1)'$ as their images.
\end{itemize}
We omit the case analysis showing that these gadgets and their terminals indeed admit only the aforementioned images. Observe that the enforcing gadgets are vital here and they also significantly simplify that analysis.

Put all together into the edge gadget (depicted in the lower part of the figure), the effect is the following:
\begin{enumerate}
  \item The vertex gadget on the left has, without loss of generality, 1 and 1' on its terminal vertices.
  \item Next, 1-gadget assures that its other terminal vertices will get 2 and 2'.
  \item The next gadget is a vertex gadget. It copies the value 2 and 2'.
  \item Now 0-1 gadget ensures that the other terminal vertices are either 2 and 2', or 3 and 3'.
  \item This is repeated once again, resulting in terminal vertices of the final vertex gadget for the vertex $v$ not having 1 and 1' as images.
\end{enumerate}
It is important to note that the edge gadget is symmetric in the sense that we could in the beginning set the image of the vertex gadget for $v$ and the rest would propagate in the similar manner, so that the terminal vertices of $u$ do not have the same ``color''.

The whole construction is again bipartite and the colors of $G$ are encoded as the pairs $i,i'$ of vertices of 4-ring, where $i \in \{1,2,3,4\}$. From the construction of the vertex gadget, we may assume that the white terminal vertices of vertex gadgets have the non-primed vertices of 4-ring as their images.

From the preceding discussion, it is clear that if $G'$ covers 4-ring, then $G$ is 4-colorable.

On the other hand, suppose that $G$ is 4-colorable. Then set the terminal vertices of vertex gadgets of $G'$ according to the coloring. We already know from the previous reductions that a vertex gadget can be labeled so that the definition of covering projection is not violated. Also, it is now clear that the vertices of ``edge gadgets'' of $G'$ can be labeled independently on each other  so that the $G'$ in the end covers 4-ring. This completes the sketch of the reduction. \qed
\end{proof}

The following observation shows that hardness for $k$-rings implies the hardness for $k$-sausages.

\begin{corollary}
\label{cor:sau_to_rings}
For every $k\ge 2$ and every $k$-sausage $S_k$, {\sc $k$-ring-Cover} $\varpropto$ {\sc $S_k$-Cover} and {\sc List-$k$-ring-Cover} \textcolor{blue}{$<_T$} {\sc List-$S_k$-Cover}. 
\end{corollary}

\begin{proof}
A graph $G$ covers $H\times K_2$ if and only if it is bipartite and covers $H$. Since bipartiteness can be tested in polynomial time, testing if $G$ covers the $k$-ring polynomially reduces to testing if $G$ covers $S_k$. 

With the list version, one has to be a bit more explicit, \textcolor{blue}{and we use a Turing reduction in this case, with two calls of {\sc List-$S_k$-Cover} in the reduction.} Let $S_k=(V,\Lambda)$, $V(K_2)=\{b,w\}$ and let the $k$-ring be denoted by $R_k$, with $V(R_k)=\{(u,\alpha):u\in V, \alpha\in\{b,w\}\}$. After checking that the input graph $G$ of the {\sc List-$R_k$-Cover} problem is bipartite, let $V(G)=A\cup B$ be the bipartition of $G$. In a feasible covering projection, either the vertices of $A$ are mapped onto the vertices of $V\times\{w\}$ and the vertices of $B$ onto $V\times\{b\}$, or vice-versa. We try these two possibilities separately. For trying the former one, reduce first the lists to $L'_u=L_u\cap (V\times\{w\})$ for $u\in A$ and $L'_u=L_u\cap (V\times\{b\})$ for $u\in B$, and adjust the lists for edges accordingly. Then regard $G$ as an instance of {\sc List-$S_k$-Cover} with the lists $\widetilde{L}_u=\{x:(x,w)\in L'_u\mbox{ or }(x,b)\in L'_u\}$, and lists for edges being adjusted accordingly. It is not difficult to see that $G$ allows a covering projection onto $R_k$ that respects ${\cal L}'$ if and only if it allows a covering projection onto $S_k$ that respects $\widetilde{\cal L}$. Check the latter possibility in a similar way and conclude that $G, {\cal L}$ is a feasible instance of {\sc List-$R_k$-Cover} if and only if at least one of the $G, \widetilde{\cal L}$ instances is feasible for the corresponding {\sc List-$S_k$-Cover} problem.  
\qed \end{proof}
\textcolor{blue}{
\section{Strong Dichotomy for regular subcubic graphs}}\label{sec:cubic}

In this section we prove the list version of the Strong Dichotomy Conjecture for regular graphs of valency at most 3. 

\medskip
\begin{proof}[\textcolor{blue}{
{\em of Theorem~\ref{thm:minorresult}}}]
\textcolor{blue}{
There are two connected 1-regular graphs -- $F(1,0)$, the one-vertex graph with a single semi-edge, and $K_2$, the simple complete graph on 2 vertices. Each of them has a finite number of possible covers ($F(1,0)$ and $K_2$ are the only covers of $F(1,0)$, and $K_2$ is the only cover of $K_2$ itself), and thus the list covering problem is solvable in constant time for each of them.}

\textcolor{blue}{
There are two types of connected 2-regular graphs -- cycles and open paths (a path is {\em open} if it starts and ends with semi-edges, all inner edges are ordinary ones and all vertices have degree 2). For the sake of better understanding, we will treat one-vertex graphs separately.}

\textcolor{blue}{
{\em Case 1A - $H=F(0,1)$, the one-vertex graph with a loop.} The only candidates for the covering graph are the cycles (including loops and digons). Since $H$ has only one vertex and only one edge (the loop), the input graph list-covers this $H$ if and only if it is a cycle and all the lists are non-empty.}

\textcolor{blue}{
{\em Case 1B - $H=F(2,0)$, the one-vertex graph with 2 semi-edges incident with its vertex.} The only candidates for the covering graph are cycles of even length and open paths, and in every covering projection the mapping of the edges of the covering graph alternate between the two semi-edges of $H$. Let $x$ be the vertex of $H$ and $a,b$ be its two semi-edges.  If a cycle $G=C_{2h}=(u_1,e_1,u_2,e_2,\ldots,u_{2h},e_{2h})$ is an input graph and $L$ the input list function, we check whether $L(u_i)=\{x\}$ for all $i=1,2,\ldots,2h$, and whether $a\in L(e_i)$ for $i=1,3,\ldots,2h-1$ and    
$b\in L(e_i)$ for $i=2,4,\ldots,2h$, or  $b\in L(e_i)$ for $i=1,3,\ldots,2h-1$ and    
$a\in L(e_i)$ for $i=2,4,\ldots,2h$. The answer to the list-covering question is {\sf yes} in the affirmative case, and {\sf no} otherwise.}

\textcolor{blue}{
If an open path $G=P_n=(e_1,u_1,e_1,\ldots,e_{n-1},u_n,e_n)$ (where $e_0$ and $e_n$ are semi-edges and all other $e_i$'s are ordinary edges) is an input graph and $L$ the input list function, we check whether $L(u_i)=\{x\}$ for all $i=1,2,\ldots,n$, and whether $a\in L(e_i)$ for $i=1,3,\ldots$ and    
$b\in L(e_i)$ for $i=2,4,\ldots$, or  $b\in L(e_i)$ for $i=1,3,\ldots$ and    
$a\in L(e_i)$ for $i=2,4,\ldots$. The answer to the list-covering question is {\sf yes} in the affirmative case, and {\sf no} otherwise.}   

\textcolor{blue}{
{\em Case 2A - $H=C_{t}=(x_1,e_1,x_2,e_2,\ldots,x_{t},e_{t})$.} The only candidates for the covering graph are cycles of lengths divisible by $t$. We reject the input as infeasible if the input graph $G$ is not such a cycle. If it is $G=C_{ht}=(u_1,a_1,\ldots,u_{ht},a_{ht}$, there are $2t$ candidates for a covering projection - $u_1$ can be mapped onto one of the $t$ vertices $x_1,\ldots,x_t$ of the target graph, and the cycle itself can ``wind around" $H$ either clock-wise or counter-clock-wise. For each of these $2t$ possibilities we check whether the mapping of vertices and edges complies with the lists (e.g., if we check the mapping of $u_1$ onto $x_j$ and wind the cycle clock-wise, we check if $x_{j+i-1}\in L(u_i)$ and $e_{j+i-1}\in L(a_i)$ for every $i=1,2,\ldots,ht$, with counting in the subscripts being modulo $t$ for the vertices and edges of $H$). We accept the input as feasible if at least one of these $2t$ cases complies with the lists, and reject otherwise. The running time is $O(tht)=O(n)$, since $t$ is a constant parameter.}

\textcolor{blue}{
{\em Case 2B - $H$ is an open path $P_t=e_0,x_1,e_1,\ldots,x_t, e_t$  with $e_0$ and $e_t$ being semi-edges (incident with  $u_1$ and $u_t$, respectively).} The only candidates for the covering graph are cycles and open paths. If the input graph is a cycle $G=C_n=(u_1,a_1,u_2,a_2,\ldots,u_n,a_n)$, a necessary condition for covering $H$ is that $n$ is a multiple of $2t$. In such a case, there are $2t$ possible covering projections, uniquely determined by the mapping of $u_1$ and $u_2$ (the vertices and edges of the cycle $C_n$ must map cyclically onto $x_1,e_1,x_2,e_2,\ldots,x_t,e_t,x_t,e_{t-1},x_{t-1},\ldots,x_1,e_0,x_1,e_1,x_2,\ldots$ in this order, or in the opposite one). For each of these $2t$ possibilities, we check if the lists $L(u_i)$ and $L(a_i)$ contain the corresponding $x$'s and $e$'s.}

\textcolor{blue}{
If the input graph is an open path $G=P_n=(a_0,u_1,a_1,u_2,a_2,\ldots,u_n,a_n)$, a necessary condition for covering $H$ is that $n$ is a multiple of $t$. In such a case, there are two possible covering projections, uniquely determined by the mapping of $u_1$  (either $a_0$ is mapped onto $e_0$ and then the path is mapped as $e_0,x_1,e_1,\ldots,x_t,e_t,x_t,e_{t-1},\ldots$, or $a_0$ is mapped onto $e_t$ and then the path is mapped as $e_t,x_t,e_{t-1},\ldots,x_1,e_0,x_1,e_1,\ldots$). For each of these two possibilities, we check if the lists $L(u_i)$ and $L(a_i)$ contain the corresponding $x$'s and $e$'s.   
\qed}
\end{proof}

\begin{proof}[\textcolor{red}{ {\em of Theorem~\ref{thm:cubicdichotomy}}}]
The proof is divided into several cases, depending on the structure of $H$.

\smallskip
\noindent {\em Case 1: We have $|V(H)|=1$.} We distinguish two subcases.

\smallskip
\noindent{\em Case 1A - The graph $H$ has one semi-edge and one loop.}
\textcolor{red}{We show that in this case, {\sc List-$H$-Cover} can be solved in polynomial time for arbitrary input graphs.}
The preimage of the semi-edge should be a disjoint union of the semi-edges of the input graph $G$ and of a perfect matching on the vertices not incident to a semi-edge.
Then the remaining edges of $G$ form a spanning collection of cycles (including loops) which form the preimage of the loop.
The existence of a spanning subgraph of $G$  that is a preimage of the semi-edge can be tested in polynomial time.

If lists are present as part of the input, the situation gets a little more tricky.
We start with a preprocessing phase. We check the below conditions:
\begin{enumerate}[label=(\alph*)]
\item $G$ has a vertex or an edge with an empty list.
\item $G$ has a vertex incident to two or more semi-edges,
\item $G$ has a semi-edge whose list does not contain the semi-edge of $H$,
\item $G$ has a vertex incident to a semi-edge and an edge, whose list does not contain the loop of $H$,
\item $G$ has a vertex incident to two ordinary edges, whose lists do not contain the loop of $H$,
\item $G$ has a loop whose list does not contain the loop of $H$.
\end{enumerate}
It is clear that if any of the above conditions is satisfied, then $(G,\cal{L})$ is a no-instance. Thus we reject and quit.

Now we shall construct an auxiliary graph $G'$.
We start our construction with $G$ and perform the following steps.
\begin{enumerate}
\item If some vertex $v$ is incident to a semi-edge, then delete $v$ with all its edges.
\item If some edge $e$ does not have the semi-edge of $H$ in its list, remove $e$ from the graph.
\item If some edge $e$ does not have the loop of $H$ in the list, leave $e$, but remove all edges incident to $e$.
\end{enumerate}
Let $G'$ be the graph after the exhaustive application of steps 1, 2, and 3.
It is straightforward to verify that steps 1 and 2 ensure that the union of a perfect matching in $G'$ and the semi-edges removed in step 1.~can be a preimage of the semi-edge of $H$.
Furthermore, by step 3 we ensure that if some edge has to be mapped to the semi-edge, then it will be so.

We can verify in polynomial time if $G'$ has a perfect matching. If not, we reject and quit.
So let $M$ be a perfect matching in $G'$, and let $M'$ be the union of $M$ and the set of semi-edges removed in step 1.
Observe that the graph $G-M$ is 2-regular, in other words a disjoint union of cycles (including loops).
Furthermore, every edge of $G-M$ has the loop of $H$ in its list, this is guaranteed by step 3 and the preprocessing phase.
Thus in this case we report a yes-instance.

\smallskip
\noindent{\em Case 1B: The graph $H$ has three semi-edges.} In this case already {\sc $H$-Cover}  is NP-complete, as it is equivalent to 3-edge-colorability of cubic graphs.

\smallskip
\noindent {\em Case 2: We have $|V(H)|=2$.} If $H$ has neither loops nor semi-edges, then $H$ is a  bipartite graph formed by a triple edge between two vertices. Only bipartite graphs can cover a bipartite one. Hence a covering projection corresponds to a 3-edge-coloring of the input graph. Thus {\sc $H$-Cover} is polynomial-time solvable (every cubic bipartite graph is 3-edge-colorable), but {\sc List-$H$-Cover} is NP-complete, because \textcolor{red}{\textsc{List 3-Edge-Coloring} is NP-complete for  cubic bipartite graphs~\cite{fiala2003np}.} If $H$ has a loop or a semi-edge, then it is one of the four graphs in Figure~\ref{fig:twovertexgraphs}, and for each of these already the {\sc $H$-Cover} problem is NP-complete~\cite{n:BFHJK-MFCS}.
\begin{figure}
\centering
\includegraphics{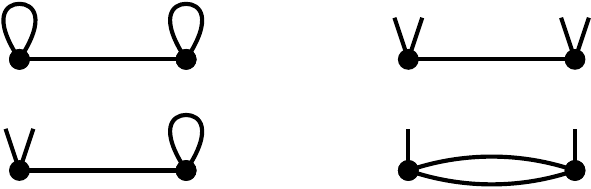}
\caption{The non-bipartite 2-vertex graphs.}
\label{fig:twovertexgraphs}
\end{figure}

\smallskip
\noindent {\em Case 3: $|V(H)| \geq 3$.} Here we split into several subcases.

\begin{itemize}
\item {\em Case 3A: The graph $H$ is acyclic.} If we shave all semi-edges off from $H$, we get a tree with at least three vertices. At least one of them has degree greater than 1, and such vertex is semi-simple in $H$. Thus {\sc List-$H$-Cover} is NP-complete by Theorem~\ref{thm:mainresult}, as we remarked earlier.

\item {\em Case 3B: The graph $H$ has a cycle of length greater than 2 which does not span all of its vertices.} Then $H$ has a vertex outside of this cycle, and thus $H$ has a semi-simple vertex and {\sc List-$H$-Cover} is NP-complete.

\item {\em Case 3C: The graph $H$ has a cycle of length greater than 2 with a chord.} Then again $H$ has a semi-simple vertex and {\sc List-$H$-Cover} is NP-complete.

\item {\em Case 3D: The graph $H$ has a cycle of length greater than 2, but none of the previous cases apply.} Then $H$ is the $k$-ring for some $k\ge 2$. If $k=2$, {\sc 2-ring-Cover} is NP-complete by~\cite{n:BFHJK-MFCS}. If $k \neq 2^{\alpha}$ for every $\alpha \geq 1$, {\sc $k$-ring-Cover} is NP-complete by Theorem~\ref{thm:rings}. In the case of $k=2^{\alpha}$ with $\alpha\ge 3$, the {\sc List-$k$-ring-Cover} problem is NP-complete by Theorem~\ref{thm:ringspowersoftwo}. For $k=4$, {\sc 4-ring-Cover} is NP-complete by Theorem~\ref{thm:4ring}.

\item {\em Case 3E - $H$ has a cycle, but all cycles are of length one or two.} 
If $H$ has a semi-simple vertex, then {\sc List-$H$-Cover} is NP-complete by Theorem~\ref{thm:mainresult}.
If $H$ has no semi-simple vertex, then $H$ is a $k$-sausage for some $k\ge 2$. The NP-completeness of {\sc List-$H$-Cover} follows from Case 3D via Corollary~\ref{cor:sau_to_rings}. \qed
\end{itemize}
\end{proof}

\section{Concluding remarks}\label{sec:concl}

We have studied the complexity of the \textsc{List-$H$-Cover} problem in the setting of graphs with multiple edges, loops, and semi-edges for regular target graphs. We have shown in Theorem~\ref{thm:mainresult} a general hardness result under the assumption that the target graph contains at least one semi-simple vertex. It is worthwhile to note that in fact we have proved the NP-hardness for the more specific \textsc{$H$-Precovering Extension} problem, when all the lists are either one-element, or full. Actually, we proved hardness for the even more specific \textsc{Vertex $H$-Precovering Extension} version, when only vertices may come with prescribed covering projections, but all edges have full lists.

On the contrary, the nature of the NP-hard cases that appear in the characterization of the complexity of \textsc{List-$H$-Cover} of cubic graphs given by Theorem~\ref{thm:cubicdichotomy} is more varied. Some of them are NP-hard already for \textsc{$H$-Cover}, some of them are NP-hard for \textsc{$H$-Precovering Extension}, but apart from the \textsc{Vertex $H$-Precovering Extension} version in applications of Theorem~\ref{thm:mainresult}, this time we also utilize the \textsc{Edge $H$-Precovering Extension} version for the case of the bipartite 2-vertex graph formed by a triple edge between two vertices. Finally, for the cases of sausages and rings whose length is a power of two, nontrivial lists are required to make our proof technique work. 

\subsection*{Acknowledgments}
Jan Bok and Nikola Jedličková were supported by research grant GAČR 20-15576S of the Czech Science Foundation and by SVV--2020--260578 and GAUK 1580119. Jiří Fiala and Jan Kratochvíl were supported by research grant GAČR 20-15576S of the Czech Science Foundation. Pawe\l{} Rz\k{a}\.zewski was supported by the Polish National Science Centre grant no. 2018/31/D/ST6/00062.
The last author is grateful to Karolina Okrasa and Marta Piecyk for fruitful and inspiring discussions.

\subsection*{Statement on confict of interest}
We declare that the authors have no competing interests as defined by Springer, or other interests that might be perceived to influence the results and/or discussion reported in this paper.

\bibliography{bib/knizky,bib/nakryti,bib/sborniky,0-main}

\end{document}